\documentclass[twoside, 9pt, web]{ieeecolor}
\usepackage{cite}
\usepackage{generic}
\usepackage{amsmath}
\usepackage{graphicx}
\usepackage{algorithmic}
\usepackage{array}
\usepackage{interval}
\usepackage{amssymb,amsfonts,mathrsfs,mathtools}
\usepackage{hyperref}
\usepackage{booktabs,makecell}
\usepackage{pict2e}
\hypersetup{colorlinks = true,
 linkcolor = red,
 urlcolor = black,
 citecolor = blue,
 anchorcolor = blue}
\usepackage{textcomp}

\def\BibTeX{{\rm B\kern-.05em{\sc i\kern-.025em b}\kern-.08em
 T\kern-.1667em\lower.7ex\hbox{E}\kern-.125emX}}
\def\gof{\mbox{${\boldsymbol{P}}\,\#\,\boldsymbol{C}$}}
\newcommand\jw{j\omega}
\newcommand\bbkt[1]{\left\{#1\right\}}
\newcommand\sbkt[1]{\left[#1\right]}
\newcommand\rbkt[1]{\left(#1\right)}
\newcommand\ininf[2]{\langle #1\,, #2 \rangle}
\newcommand\ininfbig[2]{\left\langle #1\,, #2 \right\rangle}
\newcommand\lt{\mathcal{L}_2}
\newcommand\ltp{\mathcal{L}_{2}}
\newcommand\ltep{\mathcal{L}_{2e}}
\newcommand\ccp{\bar{\mathbb{C}}_+}
\newcommand\clp{\bar{\mathbb{C}}_-}
\newcommand\cop{{\mathbb{C}}_+}

\newcommand\sysp{\boldsymbol{P}}
\newcommand\sysg{\boldsymbol{G}}
\newcommand\srg{\mathrm{SRG}}
\newcommand\gra{\mathscr{G}}
\usepackage{tikz}
\newcommand\sysi{\boldsymbol{I}}
\newcommand\sysc{\boldsymbol{C}}
\newcommand\sysf{\boldsymbol{F}}
\newcommand\rn{{\mathbb{R}}^n}
\newcommand\RI{{\mathrm{I}}}

\newcommand\abs[1]{\left|#1\right|}
\newcommand\rep{{\rm Re}}
\newcommand\imp{{\rm Im}}
\newcommand\dom{{\rm dom}}
\newcommand{\norm}[1]{\left\lVert#1\right\rVert}
\newcommand{\inorm}[1]{\left\lVert#1\right\rVert_{\mathrm{I}}}

\newcommand{\tbt}[4]{\begin{bmatrix}#1&#2\\#3&#4\end{bmatrix}}

\newcommand{\tbo}[2]{\begin{bmatrix}#1\\#2\end{bmatrix}}

\newcommand{\stbt}[4]{\left[\begin{smallmatrix} #1&#2\\#3&#4\end{smallmatrix}\right]}
\newcommand{\stbo}[2]{\left[\begin{smallmatrix} #1\\#2\end{smallmatrix}\right]}

\newcommand{\be}{\begin{equation}}\newcommand{\ee}{\end{equation}}
\newcommand{\bex}{\begin{equation*}}\newcommand{\eex}{\end{equation*}}

\def\goftau{\mbox{${\boldsymbol{P}}\,\#\, \rbkt{\tau\boldsymbol{C}}$}}
\def\fpc{\mbox{$\boldsymbol{F}_{\sysp\# \sysc}$}}

\def\myproof{\noindent\hspace{2em}{\itshape Proof of Theorem~\ref{thm: softSG}: }}
\def\endmyproof{\hspace*{\fill}~\QED\par\endtrivlist\unskip}

\newtheorem{theorem}{Theorem}
\newtheorem{proposition}{Proposition}
\newtheorem{definition}{Definition}

\newtheorem{corollary}{Corollary}
\intervalconfig{soft open fences}

\begin{document}
\title{Soft and Hard Scaled Relative Graphs for \\ Nonlinear Feedback Stability} 
\author{Chao Chen, \IEEEmembership{Member, IEEE}, Sei Zhen Khong, \IEEEmembership{Senior Member, IEEE}, and Rodolphe Sepulchre, \IEEEmembership{Fellow, IEEE} 
\thanks{This work was supported by the European Research Council under the Advanced ERC Grant Agreement SpikyControl n. 101054323. The work of Sei Zhen Khong was supported by the National Science and Technology Council of Taiwan under grants 113-2222-E-110-002-MY3, 114-2622-8-110-00 and 114-2218-E-007-011.} 
\thanks{Chao Chen is with the Department of Electrical and Electronic Engineering, The University of Manchester, Manchester M13 9PL, United Kingdom (e-mail: chao.chen@manchester.ac.uk).}
\thanks{Sei Zhen Khong is with the Department of Electrical Engineering, National Sun Yat-sen University, Kaohsiung 80424, Taiwan (e-mail: szkhong@mail.nsysu.edu.tw).} 
\thanks{Rodolphe Sepulchre is with the Department of Electrical Engineering (ESAT-STADIUS), KU Leuven, Leuven 3001, Belgium, and the Department of Engineering, University of Cambridge, Cambridge CB2~1PZ, United Kingdom (e-mail: rodolphe.sepulchre@kuleuven.be).}}
\maketitle

\begin{abstract}
This article presents input-output stability analysis of nonlinear feedback systems based on the notion of soft and hard scaled relative graphs (SRGs). The soft and hard SRGs acknowledge the distinction between incremental positivity and incremental passivity and reconcile them from a graphical perspective. The essence of our proposed analysis is that the separation of soft SRGs or hard SRGs of two open-loop systems on the complex plane guarantees closed-loop stability. The main results generalize an existing soft SRG separation theorem for bounded open-loop systems which was proved based on interconnection properties of soft SRGs under a chordal assumption. By comparison, our analysis does not require this chordal assumption and applies to possibly unbounded open-loop systems based on their hard SRGs.
\end{abstract}

\begin{IEEEkeywords}
Scaled relative graph, robust stability, graph separation, incremental positivity, incremental passivity.
\end{IEEEkeywords}

\IEEEpeerreviewmaketitle

\section{Introduction}
Graphical tools have been central to the development of feedback control theory. Representative tools include the Bode diagram, Nyquist plot, Nichols chart, Riemann plot and root locus. These tools underlie cornerstone results in linear control analysis and synthesis, among which the Nyquist stability criterion and gain/phase robustness margins are of significant importance for both theoretical and practical use. More particularly, deriving simple graphical conditions on open-loop components to handily determine stability of closed-loop linear time-invariant (LTI) systems is meaningful, e.g., small-gain, small-phase and passivity conditions.

Over the past half-century, Zames' pioneering two-part work \cite{Zames:66, Zames:66_2} on nonlinear feedback input-output stability theory has profoundly influenced research in the systems and control community. Feedback input-output stability problems boil down to boundedness and continuity problems of well-posed feedback systems~\cite{Zames:66}. Historically, boundedness and continuity are alternatively termed finite-gain stability and incremental finite-gain stability~\cite[Sect.~3]{Desoer:75}, respectively. The latter notion is stronger but more practical since it requires that output trajectories of a system must not be critically sensitive to small changes in its input trajectories, and thereby the latter was adopted in \cite{Zames:66} as a more proper stability definition for nonlinear systems. Zames' seminal work \cite{Zames:66} developed three substantial results for continuity of feedback systems: the incremental small-gain, incremental passivity and incremental conicity theorems \cite[Ths.~1-3]{Zames:66}. For single-input single-output (SISO) LTI systems, these results all enjoy clear graphical interpretations embedded in a Nyquist gain/phase setting.

Conceptually, each theorem in \cite[Ths.~1-3]{Zames:66} boils down to a specific form of \emph{graph separation} of two open-loop systems in a topological sense. Here, a system's graph is nothing but an abstract representation of its input-output energy-bounded trajectories in the Hilbert space $\lt$. Topological graph separation has been thoroughly studied and is known to be nearly the most general condition for feedback input-output stability \cite{Teel:11}, yet also the most abstract. The corresponding literature is vast; see, e.g., \cite{Sandberg:64, Safonov:80, Verma:88, Doyle:93, Georgiou:95, Teel:96, Georgiou:97, Hilborne:22} and the references therein. When separation happens to be made in a quadratic form, the approach of integral quadratic constraints~\cite{Megretski:97, Seiler:14, Carrasco:19, Khong:21} offers a more tractable alternative for feedback stability. For a Lur'e system consisting of an LTI component and a static nonlinearity, extensive efforts have been dedicated to making nonlinear feedback analysis tractable and visualizable in the complex plane, e.g., the celebrated circle criterion and Popov criterion \cite{Zames:66_2}, \cite[Sect.~6.6]{Vidyasagar:93}. The circle criterion is regarded as a nonlinear generalization of the Nyquist criterion \cite[Sect.~4]{Zames:66_2}. In addition, the recent notions of nonlinear phase~\cite{Chen:20j} and singular angle \cite{Chen:21_angle} based on numerical ranges represent other notable efforts on extracting graphical information from nonlinear systems.
 
Recently, the notion of \emph{scaled relative graph} (SRG) of nonlinear operators defined on a Hilbert space was introduced in \cite{Ryu:21} for the purpose of analysis of optimization algorithms from a graphical perspective. This notion was later adopted in \cite{Chaffey:21j,Chaffey:22_rolled} for graphical nonlinear systems analysis and feedback continuity analysis. A system's SRG, a collection of \emph{complex scalars}, mixes its gain and phase information from an increment of the system's input-output trajectories into a polar form in $\mathbb{C}$, that is,
\begin{equation*}
 \srg(\sysp) = \left\{z\in\mathbb{C}: \abs{z} = \textrm{gain}(\Delta u, \Delta y), \angle z =\textrm{phase}(\Delta u, \Delta y) \right\},
\end{equation*}
where $\sysp = u\mapsto y$ represents a bounded nonlinear system and $\Delta (\cdot)$ the difference between two trajectories in the $\lt$ space. The SRG analysis enables graphical interpretations of nonlinear systems, which is reminiscent of the classical Nyquist analysis. Typical illustrations \cite[Prop.~1]{Chaffey:21j} include that the SRG of an incrementally positive system is contained in a closed right half-plane and that of an incrementally gain-bounded system lies in a closed disk. A new separation result for feedback continuity analysis has been proposed in \cite[Th.~2]{Chaffey:21j} via the use of SRGs under certain chordal conditions, thereby endowing the incremental small-gain and positivity theorems \cite{Zames:66} graphical understandings akin to those from the Nyquist plot viewpoint.

In this paper, motivated by \cite{Zames:66} and \cite{Chaffey:21j}, we investigate feedback input-output stability based on the notion of SRG and aim at a self-contained story for SRG separation in both $\lt$ space and $\ltep$ space. The essence of our proposed results can be distilled into one sentence: \emph{Separation of the SRGs of two open-loop systems in $\mathbb{C}$ implies closed-loop stability or continuity}. The proposed results generalize the existing SRG separation theorem \cite[Th.~2]{Chaffey:21j} in the following sense: (i) the chordal assumption used for SRG over-approximation and interconnection rules is removed and (ii) possibly unbounded open-loop systems are allowed. Our proof is \emph{self-contained} in the sense that it only exploits systems' input-output trajectories classified into three trajectory-wise cases, that is, incremental small-gain, incremental large-gain and incremental small-phase cases. 

To enrich the SRG-based systems theory, we further develop an $\ltep$-framework by proposing the so-called \emph{hard} SRG with respect to trajectories lying in $\ltep$. The $\ltep$ space is an extension of $\lt$, in which the integrals are taken from time $0$ to any finite time $T>0$. The hard-type definition stands in contrast to the original \emph{soft} SRG definition in terms of $\lt$-trajectories, where the integrals are taken from time $0$ to $\infty$. In doing so, the proposed soft and hard SRGs can respectively recover to the two existing definitions of incremental positivity and incremental passivity \cite[Ch.~VI]{Desoer:75}. A hard-type definition in general is much stronger than its corresponding soft-type. The main benefit of a hard-type definition is that the corresponding graph separation result for feedback stability is more straightforward and easier to establish when compared with its soft-type counterpart, since the latter often requires additional \emph{homotopy} arguments. Such a distinction is acknowledged throughout our proposed results. Apart from incremental passivity, other hard-type input-output notions include the hard IQC \cite{Megretski:97, Seiler:14,Carrasco:19, Khong:21} and $\mathcal{L}_{2e}$-singular angle \cite{Chen:21_angle}. Hard-type notions are more common in state-space control theory, particularly in dissipativity theory \cite{Willems:72, Hill:80}. Dissipation inequalities are mostly of the hard type: incremental dissipativity \cite{Stan:07, Sepulchre:22_incremental, Verhoek:23, Liu:23}, differential dissipativity \cite{Forni:18}, and dynamic dissipativity \cite{Khong:22}, to name a few.

The remainder of this paper is structured as follows. In Section~\ref{sec:preliminaries}, notation and preliminaries on signals and systems are provided. In Section~\ref{sec: srg}, we propose soft SRGs and hard SRGs for nonlinear systems and build their links with incremental positivity and incremental passivity. In Section~\ref{sec: main}, we establish the main results of this paper, novel conditions for feedback stability analysis via separation of soft and hard SRGs. A comparison is also made between the proposed results and the existing soft SRG separation result. Section~\ref{sec:conclusion} concludes this paper.
 
\section{Notation and Preliminaries}\label{sec:preliminaries}

Let $\mathbb{F}=\mathbb{R}$ or $\mathbb{C}$ be the field of real or complex numbers and $\mathbb{F}^n$ be the linear space of $n$-dimensional vectors over $\mathbb{F}$. Let $\bar{\mathbb{C}}\coloneqq \mathbb{C}\cup \{\infty\}$, $\bar{\mathbb{C}}_+$ be the extended complex right half-plane, and $\mathbb{R}_+$ be the nonnegative real axis. For $x, y\in \mathbb{F}^n$, denote $\ininf{x}{y}$ and $\abs{x}\coloneqq \sqrt{\ininf{x}{x}}$ as the Euclidean inner product and norm, respectively. The complex conjugate, transpose and conjugate transpose of matrices are denoted by $\overline{(\cdot)}$, $(\cdot)^\top$ and $(\cdot)^*$, respectively. Let $I$ denote the identity matrix of appropriate dimensions. The real and imaginary parts of $z\in \mathbb{C}$ are denoted by $\rep\rbkt{z}$ and $\imp \rbkt{z}$, respectively. The angle of $z\in \mathbb{C} \setminus\{0\}$ in the polar form $\abs{z}e^{j\angle z}$ is denoted by $\angle z$. If $z=0$ or $z=\infty$, $\angle z$ is undefined. The closure of a set $\mathcal{S}$ is denoted by $\mathrm{cl}~\mathcal{S}$.

Denote the Lebesgue space of all energy-bounded $\rn$-valued signals by
\bex
\textstyle \lt^n \coloneqq \left\{u\colon \mathbb{R}_+\to \rn \Big| \norm{u}_2^2 = \ininf{u}{u} \coloneqq \int_{0}^{\infty} |u(t)|^2\,\mathrm{d}t < \infty\right\},
\eex 
where the superscript $n$ is dropped when the dimension is clear from the context. For $u\in \lt$, throughout this paper the simplified notation $u=0$ is adopted for $u\overset{\text{a.e.}}{=}0$, i.e., zero almost everywhere, and all integrals are taken in the Lebesgue sense. For $T\geq 0$, define the truncation operator $\boldsymbol{\Gamma}_T$ on a signal $u\colon \mathbb{R}_+ \rightarrow \rn$ to be $ (\boldsymbol{\Gamma}_T u)(t)= u(t)$ when $t\leq T$, and $(\boldsymbol{\Gamma}_T u)(t)=0$ when $t> T$. For notational simplicity, we adopt $u_T \coloneqq \boldsymbol{\Gamma}_T u$ for any $T\geq 0$ when there is no ambiguity. Denote the extended $\lt$-space by 
\bex
\ltep \coloneqq \left\{ u\colon \mathbb{R}_+ \to \rn \mid u_T \in \lt~\forall T\geq 0 \right\}
\eex
and define the semi-inner product $\ininf{u}{v}_T \coloneqq \ininf{u_T}{v_T}$ for $u \in \ltep$, $v\in \ltep$ and $T\geq 0$. 

Given $u, v\in \ltp$, define the \emph{phase} $\theta(u, v)\in \interval{0}{\pi}$ from $u$ to $v$ by
 \be\label{eq: def_phase}
 \theta(u, v)\coloneqq \arccos \frac{\ininf{u}{v}}{\norm{u}_2\norm{v}_2}
 \ee 
 if $u,v\neq 0 $, and $\theta(u, v)\coloneqq 0$, otherwise. In parallel, define the \emph{gain} $\gamma(u, v)\in \interval{0}{\infty}$ from $u$ to $v$ by
\be\label{eq: def_gain} 
\gamma(u, v)\coloneqq \frac{\norm{v}_2}{\norm{u}_2}
\ee 
if $u \neq 0$, and $\gamma(u, v)\coloneqq \infty$, otherwise. Note that $\theta(v, u)=\theta(u, v)$ and $\gamma(v, u)=1/\gamma(u, v)$ for $u, v\neq 0$. Similarly, for signals $u, v\in \ltep$ and $T>0$, we have $u_T, v_T\in \ltp$ and denote by \begin{align*}
 \theta_T(u, v)\coloneqq \theta(u_T, v_T) \text{ and } \gamma_T(u, v)\coloneqq \gamma(u_T, v_T).
\end{align*}

An operator $\sysp\colon \ltep\to\ltep$ is said to be causal if $\boldsymbol{\Gamma}_T\sysp=\boldsymbol{\Gamma}_T\sysp\boldsymbol{\Gamma}_T$ for all $T\geq 0$. We view a system as a causal operator mapping input signals to output signals. For simplicity, it is assumed that an operator maps the zero signal to the zero signal\footnote[1]{If the assumption is not satisfied due to, e.g., nonzero initial conditions, then compensating bias terms can be added into a feedback loop; see~\cite[Sect.~2]{Zames:66}. Moreover, the assumption can be removed when only incremental properties are under consideration.}, i.e., $\sysp0=0$. Without loss of generality (WLOG), only ``square'' systems with the same number of inputs and outputs are considered, as a system may always be patched with zeros to make it ``square''. Further, these systems are assumed to be nonzero, i.e., $\sysp\neq 0$. The $\lt$-domain of $\sysp$, i.e., the set of all its input signals in $\ltp$ such that the output signals are in $\ltp$, is denoted by $\mathrm{dom}(\sysp)\coloneqq \left\{ u\in \ltp \mid \sysp u \in \ltp \right\}$.
 
Throughout this paper, we will analyze the following type of stability of an open-loop system in the finite-gain sense {\cite[Sec~2.4]{Zames:66}}.
\begin{definition}[Open-loop stability]\label{def: stability}
A causal system $\sysp:\ltep \to \ltep$ is said to be 
\textit{bounded} if $\mathrm{dom}(\sysp) = \ltp$ and 
\bex 
\norm{\sysp} \coloneqq \sup_{0\neq u \in \ltp} \gamma\rbkt{u, \sysp u} < \infty.
\eex 
 Moreover, $\sysp$ is called \emph{stable} (a.k.a. continuous) if it is bounded and 
\bex 
 \norm{\sysp}_{\RI} \coloneqq \sup_{\substack{u_1,u_2 \in \ltp, u_1\neq u_2}} \gamma\rbkt{u_1-u_2, \sysp u_1 - \sysp u_2} < \infty. 
\eex
\end{definition}

In Definition~\ref{def: stability}, we follow the terminology employed in \cite{Zames:66} that a system's stability is defined by its (Lipschitz) continuity; that is, its output increment is not critically sensitive to a small change of its input. In addition, $\inorm{\sysp}$ is called the \emph{incremental gain} (a.k.a. Lipschitz gain) of $\sysp$. It is noteworthy that the incremental gain of a causal bounded system $\sysp$ can be equivalently obtained from the $\ltep$ space. By \cite[Prop.~1.2.3]{Van:17},
 it holds that \bex \inorm{\sysp} = \sup_{\substack{u_1,u_2 \in \ltep, T>0, \norm{(u_1-u_2)_T}_2\neq 0}} \gamma_T\rbkt{u_1 - u_2, \sysp u_1 - \sysp u_2}.\eex 
 
The \emph{graph} of a causal system $\sysp:\ltep\to\ltep$ is defined by $\gra(\sysp)\coloneqq \left\{\stbo{u}{\sysp u} \in \ltp\times \ltp \right\}$.
The \emph{inverse graph} of $\sysp$ is defined as $\gra^{\dagger}(\sysp)\coloneqq \left\{\stbo{\sysp u}{u} \in \ltp \times \ltp \right\}$ swapping the order of inputs and outputs. Likewise, the \emph{extended graph} and the extended inverse graph of $\sysp$ are defined by $\gra_e(\sysp)\coloneqq \left\{\stbo{u}{\sysp u} \in \ltep \times \ltep \right\}$ and $\gra^{\dagger}_e(\sysp)\coloneqq \left\{\stbo{\sysp u}{u} \in \ltep \times \ltep \right\}$, respectively. 

\section{Soft and Hard Scaled Relative Graphs}\label{sec: srg}

In classical control theory, it is well known that the Nyquist plot of a gain-bounded SISO LTI system is contained in a closed disk and that of a passive SISO LTI system is contained in a closed right-half plane. The concept of SRG \cite{Chaffey:21j} generalizes such graphical interpretations to nonlinear systems, which was originally introduced with its soft-form tested over the $\lt$ space \cite{Ryu:21, Chaffey:21j}. In this section, we introduce a parallel notion, the hard-form of SRG tested over the $\ltep$ space. We then acknowledge the distinction between incremental positivity and incremental passivity made in \cite[p.~174]{Desoer:75} in terms of the difference between soft SRGs and hard SRGs. We begin with presenting soft SRGs and then introduce hard SRGs. Let $\Delta(\cdot)$ denote the difference between two trajectories labeled by $(\cdot)_1$ and $(\cdot)_2$, e.g., $\Delta u \coloneqq u_1 - u_2$. 

\subsection{Soft and Hard SRGs}
For a causal system $\sysp\colon \ltep \to \ltep$, define the \textit{soft} SRG of $\sysp$ by
\begin{multline}\label{eq: def_soft_srg}
 \srg(\sysp)\coloneqq \left\{\gamma(\Delta u, \Delta y) e^{\pm j \theta(\Delta u, \Delta y)}\Big|\tbo{u_1}{y_1}\in \gra(\sysp), \right. \\
 \left. \vphantom{\tbo{u_2}{y_2}} \tbo{u_2}{y_2} \in \gra(\sysp), \Delta u \neq 0, \Delta y \neq 0 \right\}
 \end{multline}
as a subset of the extended complex plane $\bar{\mathbb{C}}$, with $\theta(\cdot, \cdot)$ defined in \eqref{eq: def_phase} and $\gamma(\cdot, \cdot)$ defined in \eqref{eq: def_gain}. Each pair of input-output incremental trajectories of $\sysp$ can generate a complex scalar whose magnitude and argument contain incremental gain and incremental phase information of $\sysp$, respectively. Since $\theta(\cdot, \cdot)$ takes values in $\interval{0}{\pi}$, $\srg(\sysp)$ contains the term $e^{\pm j \theta(\cdot, \cdot)}$ which makes itself symmetric about the real axis, similarly to the classical Nyquist plot. A signed version of the soft SRG has been recently investigated in \cite{VanDenEijnden:25}. Furthermore, the definition of $\srg(\sysp)$ in \eqref{eq: def_soft_srg} is based on the graph $\gra(\sysp)$ or $u\in \dom(\sysp)$, thereby also encompassing possibly unbounded systems with $\dom(\sysp)$ being a proper subset of $\lt$. Notice that soft SRGs are not necessarily closed, e.g., see \cite[Ex.~1]{Chaffey:24}. 

The \emph{inverse soft} SRG of $\sysp$ is defined as
\begin{equation}\label{eq: def_inv_soft_srg}
\begin{aligned}
 \srg^{\dagger}(\sysp)\coloneqq \left\{\gamma(\Delta y, \Delta u) e^{\pm j \theta(\Delta y, \Delta u)} \Big|\tbo{y_1}{u_1} \in \gra^{\dagger}(\sysp), 
\right. \\
 \left. \vphantom{\tbo{y_2}{u_2}} \tbo{y_2}{u_2} \in \gra^\dagger(\sysp), \Delta u\neq 0, \Delta y\neq 0 \right\}.
 \end{aligned} 
\end{equation}
In comparison to $\srg(\sysp)$, the above $\srg^{\dagger}(\sysp)$ exploits the inverse graph $\gra^{\dagger}(\sysp)$ by \emph{swapping} the role of input and output of $\sysp$. We refer the reader to \cite[Prop.~1]{Chaffey:21j} for illustrations of soft SRGs of some typical bounded systems and to \cite[Th.~4]{Chaffey:21j} for the connection between soft SRGs and Nyquist plots of SISO LTI systems. 
 
The definition in~\eqref{eq: def_soft_srg} is with respect to $\lt$-trajectories lying in $\gra(\sysp)$ which is classified as a \textit{soft}-type. By contrast, a \textit{hard}-type definition can be proposed by using $\ltep$-trajectories contained in the extended graph $\gra_e(\sysp)$. The terminology of soft and hard is borrowed from the theory of soft and hard IQCs \cite{Megretski:97, Carrasco:19, Khong:21}. Specifically, the \emph{hard} SRG of a causal system $\sysp:\ltep \to \ltep$ is defined by
\begin{multline}\label{eq: def_hard_srg}
 \srg_e(\sysp)\coloneqq \left\{\gamma_T(\Delta u, \Delta y) e^{\pm j \theta_T(\Delta u, \Delta y)} \Big| \tbo{u_1}{y_1} \in \gra_e(\sysp), \right.\\
 \left. \tbo{u_2}{y_2} \in \gra_e(\sysp), (\Delta u)_T \neq 0, (\Delta y)_T \neq 0, T>0 \vphantom{e^{j\theta}} \right\}.
 \end{multline}

The \textit{inverse hard} SRG of $\sysp$, $\srg^{\dagger}_e(\sysp)$, can be defined in a similar manner to \eqref{eq: def_inv_soft_srg}, based on the \emph{extended inverse graph} $\gra_e^{\dagger}(\sysp)$. It is worth noting that definition~\eqref{eq: def_hard_srg} can also accommodate possibly unbounded systems like an LTI integrator with transfer function $\frac{1}{s}I$. One of our motivations for placing the soft SRG \eqref{eq: def_soft_srg} and the hard SRG \eqref{eq: def_hard_srg} on the same footing is to recover two existing notions, incremental positivity and incremental passivity, respectively, as explained later. 

\begin{table}[htb]
 \begin{center}
 \begin{tabular}{ccc}
 \toprule[1pt] & Soft $\srg(\sysp)$ & Hard $\srg_e(\sysp)$ \\
 \specialrule{0em}{2pt}{2pt}
 \hline
 \specialrule{0em}{2pt}{2pt}
 Trajectories & $\lt$ & $\ltep$ \\
 \specialrule{0em}{2pt}{2pt}
 Connections & Incremental positivity & Incremental passivity \\
 \bottomrule[1pt]
 \end{tabular}
 \end{center}
 \caption{Soft and hard SRGs of a causal system $\sysp$.}\label{tab:graph}
\end{table}

In summary, for a causal system $\sysp$, two types of scaled relative graphs have been defined, as illustrated by Table~\ref{tab:graph}.
 
\subsection{Characterizations of Incremental Positivity and Incremental Passivity via Soft and Hard SRGs}\label{sec: incrementalpassivity}
The theory of positivity and passivity~\cite[Ch.~VI]{Desoer:75} has been one of the cornerstones of input-output nonlinear control theory. We present two existing definitions~\cite[Sect.~VI.4]{Desoer:75} to characterize incremental positivity or incremental passivity of a system. 
 \begin{definition}[Positivity and passivity]\label{def: passive}
 A causal system $\sysp= u\mapsto y:\ltep\to\ltep$ is said to be 
\begin{enumerate}
 \renewcommand{\theenumi}{\textup{(\roman{enumi})}}\renewcommand{\labelenumi}{\theenumi}
 \item \label{item: l2_inc_passive_def} \emph{incrementally positive} if 
 \be \label{eq: l2_inc_passive_def}
 \ininf{\Delta u}{\Delta y}\geq 0 \quad \forall u_1, u_2\in \dom(\sysp);
 \ee 
 \item \label{item: l2e_inc_passive_def} \emph{incrementally passive} if 
 \be \label{eq: l2e_inc_passive_def}
 \ininf{(\Delta u)_T}{(\Delta y)_T}\geq 0 \quad \forall T>0, u_1, u_2\in \ltep.
 \ee
 Furthermore, $\sysp$ is said to be 
 \item \label{item: l2_inc_spassive_def} \emph{strictly incrementally positive} if there exist $\delta, \epsilon>0$ such that 
\be \label{eq: l2_inc_spassive_def}
\ininf{\Delta u}{ \Delta y}\geq \delta \norm{\Delta u}_2^2 + \epsilon \norm{\Delta y}_2^2 \quad \forall u_1, u_2\in \dom(\sysp);
\ee
 \item \label{item: l2e_inc_spassive_def} \emph{strictly incrementally passive} if there exist $\delta, \epsilon>0$ such that 
 \be\label{eq: l2e_inc_spassive_def}
 \ininf{(\Delta u)_T}{(\Delta y)_T}\geq \delta \norm{(\Delta u)_T}_2^2 + \epsilon \norm{(\Delta y)_T}_2^2
 \ee for all $T>0$ and $u_1, u_2\in \ltep$.
\end{enumerate}
 \end{definition}

As per Definition~\ref{def: passive}, the positivity is tested over $\lt$ while the passivity over $\ltep$ \cite[Ch.~VI]{Desoer:75}.
We now reformulate the difference between the incremental positivity and incremental passivity based on soft and hard SRGs. Explicitly, an equivalent characterization of Definition~\ref{def: passive}\ref{item: l2_inc_passive_def}-\ref{item: l2e_inc_passive_def} can be given as follows:
\begin{align*}
 \eqref{eq: l2_inc_passive_def} \Leftrightarrow \srg(\sysp)\subset \ccp\setminus\{0\} \text{ and } \eqref{eq: l2e_inc_passive_def} \Leftrightarrow\srg_e(\sysp)\subset \ccp\setminus\{0\}.
\end{align*} 
The strict version in Definition~\ref{def: passive}\ref{item: l2_inc_spassive_def}-\ref{item: l2e_inc_spassive_def} by intuition should be linked with a certain SRG in an open right half-plane $\cop$. This is indeed correct and can be depicted more precisely as follows. Given $\epsilon, \delta>0$, we define the truncated disk sector $\mathcal{D}(\delta, \epsilon)$ shown in Fig.~\ref{fig:passive_phase} by
\begin{equation}\label{eq: sectored_disk}
\begin{aligned}
\mathcal{D}(\delta, \epsilon)\coloneqq \{z\in \cop \mid \abs{\angle z}\leq \arccos 2\sqrt{\delta \epsilon}, 
\abs{z} \in \interval[open left]{0}{\textstyle \frac{1}{\epsilon}}, \rep(z)\geq \delta\}.
\end{aligned}
 \end{equation}
 
 \begin{figure}[htb]
 \vspace{-2mm}
\begin{center}
 \begin{tikzpicture}
 \pgfdeclareverticalshading{cone}{100bp}{
 color(0bp)=(gray!50); 
 color(100bp)=(white)}
 \draw[->] (-1,0) -- (3.2,0) node[right] {$\rep$}; 
 \draw[->] (0,-1.2) -- (0,1.5) node[above] {$\imp$}; 
 \fill[gray!50] (0.2,0.1) -- (2,1) arc[start angle = 30, end angle = -30, radius = 2cm] -- cycle;
 \fill[gray!50] (0.2,-0.1) -- (2,1) arc[start angle = 30, end angle = -30, radius = 2cm] -- cycle;
 \fill[gray!50] (0.2,-0.1) rectangle (0.4,0.1);
 \draw [dashed](0.8,0) arc[start angle=0, end angle=30, radius=0.8];
 \node at (1, 0.2) {$\theta$};
 \draw (0.2,0.1) -- (2,1);
 \draw (0.2,-0.1) -- (2,-1);
 \draw (0.2,0.1) -- (0.2,-0.1);
 \draw (2,-1) arc[start angle=-30, end angle = 30, radius = 2cm];
 \node at (2.7, -0.3) {${1}/{\epsilon}$};
 \node at (2.2, 1.5) { $\theta = \arccos 2\sqrt{\delta \epsilon}$}; 
 \node at (0.2, 0.3) {$\delta$}; 
\node at (1.6, -0.3) {$\mathcal{D}(\delta, \epsilon)$}; 
 \end{tikzpicture}
\end{center}
\caption{An outer bound (i.e., the gray region $\mathcal{D}(\delta, \epsilon)$) of the soft SRG of a strictly incrementally positive system with indices $\delta, \epsilon>0$.} \label{fig:passive_phase}
 \end{figure}

\begin{proposition}\label{prop:strictpassive}
 For a strictly incrementally positive system $\sysp$ in \eqref{eq: l2_inc_spassive_def} with $\delta, \epsilon>0$, it holds that 
 $\srg(\sysp) \subset \mathcal{D}(\delta, \epsilon)$. Similarly, 
for a strictly incrementally passive system $\sysp$ in \eqref{eq: l2e_inc_spassive_def} with $\delta, \epsilon>0$, it holds that
 $\srg_e(\sysp) \subset \mathcal{D}(\delta, \epsilon)$.
\end{proposition} 
 \begin{proof}
For brevity, we only establish the proof for $\srg(\sysp) \subset \mathcal{D}(\delta, \epsilon)$. Firstly, it follows from \eqref{eq: l2_inc_spassive_def} that for all $u_1, u_2\in \dom(\sysp)$ such that $\Delta u\neq 0$ and $\Delta y \neq 0$, we have
 \begin{align*}
 \frac{\ininf{\Delta u}{\Delta y}}{\norm{\Delta u}_2\norm{\Delta y}_2} \geq \frac{\delta \norm{\Delta u}_2^2 + \epsilon \norm{\Delta y}_2^2}{\norm{\Delta u}_2\norm{\Delta y}_2} \geq 2\sqrt{\delta \epsilon},
 \end{align*}
 where the last inequality uses that $a^2+b^2\geq 2ab$ for any $a, b>0$.
 This implies that $\srg(\sysp)\subset \{z\in \ccp\mid \abs{\angle z} \leq \arccos 2\sqrt{\delta \epsilon}<\pi/2, z\neq 0 \}.$
 Secondly, for all $u_1, u_2\in \dom(\sysp)$, by \eqref{eq: l2_inc_spassive_def} we obtain
$\epsilon \norm{\Delta y}_2^2 \leq \ininf{\Delta u}{\Delta y} \leq \norm{\Delta u}_2\norm{\Delta y}_2.$
 Together with $\Delta u\neq 0$ and $\Delta y \neq 0$, we then have 
 \bex 
 0<\frac{\norm{\Delta y}_2}{\norm{\Delta u}_2} \leq \frac{1}{\epsilon} \iff \srg(\sysp)\subset \{z\in \mathbb{C}\mid \abs{z}\in \interval[open left]{0}{1/\epsilon}\}.
 \eex 
Thirdly, by \eqref{eq: l2_inc_spassive_def} we also obtain that $\ininf{\Delta u}{ \Delta y}\geq \delta \norm{\Delta u}_2^2$ which is equivalent to $\srg(\sysp)\subset \{ z\in \ccp\mid \rep(z)\geq\delta\}$ by \eqref{eq: def_soft_srg}. Clearly, $\srg(\sysp)$ must be contained in the intersection of the three regions, i.e., $\srg(\sysp) \subset \mathcal{D}(\delta, \epsilon)$ in \eqref{eq: sectored_disk}. 
\end{proof}

\subsection{Connections Between Soft and Hard SRGs}\label{section: soft_hard_connection}

For a causal bounded system $\sysp$, it is known from \cite[p.~200]{Desoer:75} that the incremental positivity over $\lt$ in \eqref{eq: l2_inc_passive_def} and incremental passivity over $\ltep$ in \eqref{eq: l2e_inc_passive_def} are equivalent. We have also seen in Section~\ref{sec:preliminaries} that $\inorm{\sysp}$ can be equivalently defined on $\lt$ or $\ltep$. This inspires us to explore a connection between soft SRGs and hard SRGs, which is a lot \emph{trickier}. It is easy to conclude the following direction:

 \begin{proposition}\label{prop: soft_hard}
 For a causal system $\sysp:\ltep \to \ltep$, it holds that \bex \srg(\sysp)\subset \mathrm{cl}~\srg_e(\sysp).\eex
 \end{proposition}
 \begin{proof}
For a causal system $\sysp$, any trajectory $\stbo{u}{y}\in \gra(\sysp)$ must be a trajectory in $\gra_e(\sysp)$ as $u, y\in \ltep$. Thus, the complex number $z\in \srg(\sysp)$ generated by $\stbo{u_1}{y_1}, \stbo{u_2}{y_2}\in \gra(\sysp)$ is also a point in the closure of $\srg_e(\sysp)$ as $T\to \infty$ in~\eqref{eq: def_hard_srg}. 
 \end{proof}
 
In general, $\srg(\sysp)$ is only a \emph{proper} subset of $\mathrm{cl}~\srg_e(\sysp)$. To acknowledge this, consider an LTI integrator $\sysp$ with transfer function $\frac{1}{s}$ whose domain can be characterized as $\dom(\sysp)=\frac{s}{s+1}\mathcal{L}_2$. The soft SRG of $\sysp$ is straightforward to obtain. For all $u\in \dom(\sysp)$, according to the input-output mapping $y(t)=\int_0^t u(\tau) \,\mathrm{d}\tau$ and by the strict causality of $\sysp$, we have 
\begin{align*}
 \ininf{u}{y}&=\int_0^{\infty} u(t) y(t) \, \mathrm{d}t = \lim_{T\to \infty}\int_{y(0)}^{y(T)} y \, \mathrm{d}y\\
 &=\frac{1}{2} \lim_{T\to\infty} \abs{y(T)}^2 - \frac{1}{2}\abs{y(0)}^2=0,
\end{align*}
where the last equality follows from $y(0)=0$ and that $y(T) \to 0$ as $T\to \infty$ since $u, y\in \lt$. Hence, we obtain $\srg(\sysp)=j\mathbb{R} \setminus\{0\}$.
The hard SRG is however markedly different. For all $u\in \ltep$ and $T>0$, we have 
\begin{align*}
 \ininf{u_T}{y_T} = \int_{0}^{T} u(t) y(t)\,\mathrm{d} t = \frac{1}{2} \abs{y(T)}^2 \geq 0
\end{align*}
since $y(T)$ can take possibly nonzero values for $T>0$. It is then easy to conclude $\srg_e(\sysp)=\ccp\setminus{\{0\}}$. The distinction between the above soft and hard SRGs can be clearly understood from sketching the Nyquist plot of $\frac{1}{s}$. A semi-circle detour around the unstable pole $s=0$ is required in constructing the Nyquist intended contour for $\frac{1}{s}$, which generates a phase-shift of $180^\circ$ in $\ccp$ when $s$ travels along the detour. For the soft SRG, the input-output trajectories are both restricted to be in $\lt$ so that the unbounded mode due to the pole $s=0$ will no longer be excited. In this case, the effect of the detour is ``neglected'' in the sketch and thereby $\{ \frac{1}{\jw}\mid \omega \in \mathbb{R} \setminus \{0\} \}$ is obtained. On the contrary, for the hard SRG, the unbounded mode is always excited since all the $\ltep$-trajectories are involved, which results in the phase-shift of $180^\circ$ in the sketch: $\{{\epsilon^{-1} e^{-j\alpha}\mid \alpha\in \interval{-\frac{\pi}{2}}{\frac{\pi}{2}}, \epsilon>0} \}$. Thus, $\srg_e(\sysp)=\ccp\setminus{\{0\}}$ meets our expectation.

The relationship between the hard SRG and the Nyquist plot of SISO LTI systems is nontrivial, possibly different from that for the soft SRG in~\cite[Th.~4]{Chaffey:21j}, and deserves further investigation.

\section{Soft and Hard SRG Separation}\label{sec: main}
In this section, we present two main results of this paper, stability analysis of feedback systems based on the use of soft and hard SRGs in order. The results show that separation of the two SRGs of open-loop systems in $\bar{\mathbb{C}}$ guarantees feedback stability. 

 \begin{figure}[htb]
 \centering
 \setlength{\unitlength}{1mm}
 \begin{picture}(50,25)
 \thicklines \put(0,20){\vector(1,0){8}} \put(10,20){\circle{4}}
 \put(12,20){\vector(1,0){8}} \put(20,15){\framebox(10,10){$\sysp$}}
 \put(30,20){\line(1,0){10}} \put(40,20){\vector(0,-1){13}}
 \put(38,5){\vector(-1,0){8}} \put(40,5){\circle{4}}
 \put(50,5){\vector(-1,0){8}} \put(20,0){\framebox(10,10){$\sysc$}}
 \put(20,5){\line(-1,0){10}} \put(10,5){\vector(0,1){13}}
 \put(5,10){\makebox(5,5){$y_2$}} \put(40,10){\makebox(5,5){$y_1$}}
 \put(0,20){\makebox(5,5){$d_1$}} \put(45,0){\makebox(5,5){$d_2$}}
 \put(13,20){\makebox(5,5){$u_1$}} \put(32,0){\makebox(5,5){$u_2$}}
 \end{picture}\caption{A feedback system $\gof$.} \label{fig:feedback}
 \end{figure}

Consider a positive feedback system in Fig.~\ref{fig:feedback}, where $\boldsymbol{P}\colon \ltep^n \rightarrow \ltep^n$ and $\boldsymbol{C}\colon \ltep^n \rightarrow \ltep^n$ are two causal systems, $d_1$ and $d_2$ are external trajectories, and $u_1, u_2, y_1$ and $y_2$ are internal trajectories. These trajectories are related by the following feedback equations:
\be\label{eq:feedback}
 u=d+\tbt{0}{\boldsymbol{I}}{\boldsymbol{I}}{0}y \text{ and } y=\tbt{\boldsymbol{P}}{0}{0}{\boldsymbol{C}}u,
\ee
where $u=\left[u_1^\top~u_2^\top\right]^\top$, $d=\left[{d_1^\top}~{d_2^\top}\right]^\top$ and $y=\left[{y_1^\top}~{y_2^\top}\right]^\top$. Let $\gof$ denote such a feedback system and we define the following feedback mapping:
 \bex
 \fpc\coloneqq \tbt{\sysi}{-\sysc}{-\sysp}{\sysi} = u \mapsto d: \ltep^{2n} \rightarrow \ltep^{2n}.
 \eex 
Throughout this paper, all feedback systems are assumed to be well-posed in the following sense. 
 
\begin{definition}[Well-posedness]
A feedback system $\gof$ is called \emph{well-posed} if $\fpc$ has a causal inverse on $\ltep^{2n}$.
\end{definition}

The input-output stability of $\gof$ is defined as follows.
\begin{definition}[Feedback stability]\label{def:feedback_stability}
 A well-posed feedback system $\gof$ is said to be \emph{stable} if it is bounded and $\left\|(\fpc)^{-1}\right\|_{\RI}<\infty$.
\end{definition}

We are ready to present the main results on stability analysis of $\gof$ by using soft and hard SRGs sequentially. The results can be summarized in one sentence: 
 
\emph{Feedback stability of $\gof$ is guaranteed if there is no intersection between the SRG of $\sysp$ and the inverse SRG of $\sysc$.}
 
The results can be regarded as a graphical statement of the classical topological graph separation \cite{Zames:66, Teel:96, Georgiou:97, Teel:11}. The use of soft and hard notions will entail different separation conditions and the latter is simpler. We begin with hard-type separation for simplicity.

\subsection{Hard SRG Separation}
The first result below establishes a feedback stability condition based on separation of hard SRGs in $\bar{\mathbb{C}}$.

\begin{theorem}[Hard SRG Separation]\label{thm: hardSG}
Consider a well-posed feedback system $\gof$ with $d_2=0$, where $\sysp$ is stable. Then, $\gof$ with $d_2=0$ is stable if 
\be\label{eq: hard_SRG_separation}
 \inf_{\substack{z_{1} \in \srg_e(\sysp), z_{2} \in \srg^\dagger_e(\sysc)}} \abs{ z_1 - z_2} >0.
\ee
\end{theorem}

\begin{proof} 
Consider two groups of trajectories for $\gof$ in Fig.~\ref{fig:feedback}, $u_1, u_2, y_1, y_2, d_1, d_2\in \mathcal{L}_{2e}$ with ${d}_2=0$ and $\bar{u}_1, \bar{u}_2, \bar{y}_1, \bar{y}_2, \bar{d}_1, \bar{d}_2\in \mathcal{L}_{2e}$ with $\bar{d}_2=0$. For notational brevity, the increment between $u_1$ and $\bar{u}_1$ is abbreviated as $\Delta u_1\coloneqq u_1-\bar{u}_1$. Similarly, $\Delta u_2, \Delta y_1, \Delta y_2$ and $\Delta d_1$ are adopted. Let $\Delta u \coloneqq u -\bar{u}$, where $u\coloneqq [u_1^\top~u_2^\top]^\top$ and $\bar{u}\coloneqq [{\bar{u}_1^\top}~{\bar{u}_2^\top}]^\top$. For $T>0$, the semi-norm $\norm{(\cdot)_T}_2$ of an $\ltep$ signal is shortened as $\norm{(\cdot)}_T$ without ambiguity. In a feedback loop, note that if $\Delta u_1 = 0$, then $\Delta y_1 = 0$ (the dependence of $(\cdot)_T$ on $T>0$ is omitted here). It follows from $d_2=0, \bar{d}_2=0$ and \eqref{eq:feedback} that $\Delta u_2=0$ and thus $\Delta y_2 = 0$. This further implies that $\Delta d_1 = \Delta u_1 =0$. Hence, the feedback loop with $d_2=0$ is stable. In addition, in the case where $\Delta y_2=0$, then $\Delta u_1 = \Delta d_1$ and the feedback loop is also stable due to the open-loop stability of $\sysp$. It thus suffices to consider only nonzero increments in the following stability analysis. WLOG, in what follows we can consider only points in hard SRGs in the closed upper half-plane due to the symmetric nature in \eqref{eq: def_hard_srg}.

By hypothesis~\eqref{eq: hard_SRG_separation}, for all $z_1\in \srg_e(\sysp)$ and $z_2\in \srg^{\dagger}_e(\sysc)$, we have $\abs{z_1-z_2}\geq \delta>0$ for some constant $\delta>0$. This means that at least one of the following two cases has to hold:
 \begin{enumerate}
 \renewcommand{\theenumi}{\textup{(\roman{enumi})}}\renewcommand{\labelenumi}{\theenumi}
 \item $\abs{\abs{z_1} -\abs{z_2}}\geq \epsilon$;\label{item: thm_proof_gain} 
 \item $\abs{\angle z_1 -\angle z_2} \geq \epsilon$, \label{item: thm_proof_angle} 
 \end{enumerate}
where $\epsilon>0$ is a constant that can be determined from given $\delta$. For each case, we show in the following that the internal increment $\Delta u$ is bounded by the external increment $\Delta d_1$, that is, $\norm{\Delta u}_T \leq c \norm{\Delta d_1}_T$ for all $T>0$, where $c>0$ is a constant.
 
Case \ref{item: thm_proof_gain}: $\abs{\abs{z_1} -\abs{z_2}}\geq \epsilon$. Firstly, consider $\abs{z_2}-\abs{z_1}\geq \epsilon$. By $z_1\in \srg_e(\sysp)$ in definition~\eqref{eq: def_hard_srg} and $z_2\in \srg^{\dagger}_e(\sysc)$, for all $u_1, u_2, \bar{u}_1,$ $\bar{u}_2\in \ltep$ such that $\abs{z_1} \leq \abs{z_2} - \epsilon$ holds, we have
 \begin{align}\label{thm:thm1_eq1}
 \frac{\norm{\Delta y_1 }_T}{\norm{\Delta u_1}_T} \frac{\norm{\Delta y_2}_T}{\norm{\Delta u_2}_T} \leq 1- \epsilon \frac{\norm{\Delta y_{2}}_T}{\norm{\Delta u_{2}}_T} < 1
 \end{align}
for all $T>0$, where $(\Delta u_1)_T\neq 0$, $(\Delta u_2)_T\neq 0$ and $(\Delta y_2)_T\neq 0$. Note that according to \eqref{thm:thm1_eq1}, for $1- \epsilon \frac{\norm{\Delta y_{2}}_T}{\norm{\Delta u_{2}}_T}$ to be arbitrarily close to $1$, the term $\frac{\norm{\Delta y_{2}}_T}{\norm{\Delta u_{2}}_T}$ needs to be arbitrarily close to $0$, which means that $ \frac{\norm{\Delta y_1 }_T}{\norm{\Delta u_1}_T} \frac{\norm{\Delta y_2}_T}{\norm{\Delta u_2}_T} \to 0$ since $\sysp$ is stable, i.e., $\frac{\norm{\Delta y_1 }_T}{\norm{\Delta u_1}_T} \leq \norm{\sysp}_\mathrm{I}$. Hence, $\frac{\norm{\Delta y_1 }_T}{\norm{\Delta u_1}_T} \frac{\norm{\Delta y_2}_T}{\norm{\Delta u_2}_T}$ is strictly less than $1$, independently of $\Delta u_1, \Delta u_2, \Delta y_1, \Delta y_2$ and $T$. WLOG, assume that 
\begin{equation}\label{eq: thm1_eq1}
 \frac{\norm{\Delta y_{1}}_T}{\norm{\Delta u_{1}}_T}\leq \alpha_1(\Delta u_1, \Delta y_1, T) \text{ and } \frac{\norm{\Delta y_{2}}_T}{\norm{\Delta u_{2}}_T}\leq \alpha_2(\Delta u_2, \Delta y_2, T)
\end{equation}
for all $T>0$, where $\alpha_1(\Delta u_1, \Delta y_1, T)>0$ and $\alpha_2(\Delta u_2, \Delta y_2, T)>0$ (shortened as $\alpha_1(\cdot)$ and $\alpha_2(\cdot)$ below). From the analysis above, there exists a constant $\alpha$, independent of $\Delta u_1, \Delta u_2, \Delta y_1, \Delta y_2$ and $T$, such that
\begin{align}\label{eq: thm1_product_alpha}
 \alpha_1(\cdot)\alpha_2(\cdot) \leq \alpha <1.
\end{align}
Using the feedback equations $u_1=d_1+y_2$ and $u_2=y_1$, it follows from \eqref{eq: thm1_eq1} and the triangle inequality that 
\begin{align*}
 \norm{\Delta u_{1}}_T&\leq \norm{\Delta d_{1}}_T +\norm{\Delta y_{2}}_T \leq \norm{\Delta d_{1}}_T + \alpha_2(\cdot)\norm{\Delta u_{2}}_T,\\
 \norm{\Delta u_{2}}_T&= \norm{\Delta y_{1}}_T\leq \alpha_1(\cdot)\norm{\Delta u_{1}}_T, 
\end{align*}
for all $T>0$. Therefore, we have
 \begin{align*}
 \norm{\Delta u_{1}}_T\leq \norm{\Delta d_{1}}_T + \alpha_2(\cdot)\alpha_1(\cdot)\norm{\Delta u_{1}}_T
 \end{align*}
for all $T>0$. Since $\alpha_1(\cdot)\alpha_2(\cdot) \leq \alpha<1$ in \eqref{eq: thm1_product_alpha}, we arrive at
 \begin{align}
 \norm{\Delta u_{1}}_T&\leq \frac{1}{1-\alpha}\norm{\Delta d_{1}}_T, \label{eq: thm1_oneside_gain_1}\\ 
 \norm{\Delta u_{2}}_T&\leq {\norm{\sysp}_{\mathrm{I}} \norm{\Delta u_{1}}_T \leq \frac{\norm{\sysp}_{\mathrm{I}}}{1-{ \alpha}}}\norm{\Delta d_{1}}_T \notag 
 \end{align}
 for all $T>0$. Since $\norm{\Delta u}_T \leq \norm{\Delta u_1}_T + \norm{\Delta u_2}_T$, then for all $u_1, u_2, \bar{u}_1, \bar{u}_2 \in \ltep$ such that \eqref{thm:thm1_eq1} holds, $\Delta u$ is finite-incremental-gain bounded by $\Delta d_1\in \ltep$; namely, for all $T>0$, 
 \be\label{eq: thm1_eq5}
\norm{\Delta u}_T \leq \frac{1+\norm{\sysp}_{\mathrm{I}}}{ 1- \alpha}\norm{\Delta d_1}_T\eqqcolon c_1 \norm{\Delta d_1}_T.
 \ee 

Secondly, consider $\abs{z_1}-\abs{z_2}\geq \epsilon$. By analogy with the above arguments, for all $u_1, u_2, \bar{u}_1, \bar{u}_2\in \ltep$ such that $\abs{z_1}\geq \abs{z_2} + \epsilon$ holds, we have
 \begin{align}\label{thm:thm1_eq2}
 \frac{\norm{\Delta y_{1}}_T}{\norm{\Delta u_{1}}_T} \frac{ \norm{\Delta y_{2}}_T}{ \norm{\Delta u_{2}}_T} \geq 1 + \epsilon{\frac{ \norm{\Delta y_{2}}_T}{ \norm{\Delta u_{2}}_T}} > 1 
 \end{align} 
 for all $T>0$, where $(\Delta u_1)_T\neq 0$, $(\Delta u_2)_T\neq 0$ and $(\Delta y_2)_T\neq 0$. For $1+ \epsilon \frac{\norm{\Delta y_{2}}_T}{\norm{\Delta u_{2}}_T}$ to be arbitrarily close to $1$, $\frac{\norm{\Delta y_{2}}_T}{\norm{\Delta u_{2}}_T}$ needs to be arbitrarily close to $0$, implying that $ \frac{\norm{\Delta y_1 }_T}{\norm{\Delta u_1}_T} \frac{\norm{\Delta y_2}_T}{\norm{\Delta u_2}_T} \to 0$, contradicting \eqref{thm:thm1_eq2}. Hence, $\frac{\norm{\Delta y_1 }_T}{\norm{\Delta u_1}_T} \frac{\norm{\Delta y_2}_T}{\norm{\Delta u_2}_T}$ must be uniformly strictly greater than $1$. WLOG, assume that
 \begin{align}\label{eq: thm1_eq2}
 \frac{\norm{\Delta y_{1}}_T}{\norm{\Delta u_{1}}_T}\geq \beta_1(\Delta u_1, \Delta y_1, T) \text{ and } \frac{\norm{\Delta y_{2}}_T}{\norm{\Delta u_{2}}_T}\geq \beta_2(\Delta u_2, \Delta y_2, T)
 \end{align}
 for all $T>0$, 
 where $\beta_1(\Delta u_1, \Delta y_1, T)>0, \beta_2(\Delta u_2, \Delta y_2, T)>0$
\begin{align*}
 \text{and } \beta_1(\cdot)\beta_2(\cdot)\geq \beta>1
\end{align*}
with $\beta$ being a constant owing to the analysis above. 
 Using the relations $y_2=u_1-d_1$ and $y_1=u_2$, it follows from \eqref{eq: thm1_eq2} that
 \begin{align*}
 \beta_1(\cdot)\norm{\Delta u_{1}}_T&\leq \norm{\Delta y_{1}}_T = \norm{\Delta u_{2}}_T,\\ 
 \beta_2(\cdot)\norm{\Delta u_{2}}_T&\leq \norm{\Delta y_{2}}_T\leq \norm{\Delta u_{1}}_T + \norm{\Delta d_{1}}_T
 \end{align*}
 for all $T>0$. Therefore, we immediately have
 \begin{align*}
 \norm{\Delta u_{1}}_T\leq \beta_1(\cdot)^{-1} \beta_2(\cdot)^{-1} \rbkt{ \norm{\Delta u_{1}}_T +\norm{\Delta d_{1}}_T}
 \end{align*}
 for all $T>0$.
 Since $\beta_1(\cdot)\beta_2(\cdot) \geq \beta >1$, we arrive at
 \begin{align}
 \norm{\Delta u_{1}}_T&\leq \frac{1}{\beta-1}\norm{\Delta d_{1}}_T, \label{eq: thm1_oneside_gain_2}\\ 
 \norm{\Delta u_{2}}_T&\leq \norm{\sysp}_{\mathrm{I}} \norm{\Delta u_{1}}_T \leq \frac{ \norm{\sysp}_{\mathrm{I}}}{\beta-1}\norm{\Delta d_{1}}_T \notag 
 \end{align}
 for all $T>0$. 
 This gives that for all $u_1, u_2, \bar{u}_1, \bar{u}_2\in \ltep$ such that \eqref{thm:thm1_eq2} holds, for all $T>0$, we have 
 \be\label{eq: thm1_eq6}
 \norm{\Delta u}_T \leq \frac{1+\norm{\sysp}_{\mathrm{I}}}{\beta -1}\norm{\Delta d_{1}}_T\eqqcolon c_2 \norm{\Delta d_1}_T.
 \ee 
 
Case \ref{item: thm_proof_angle}: $\abs{\angle z_1 -\angle z_2} \geq \epsilon$. Firstly, consider $\angle{z_2}-\angle{z_1}\geq \epsilon$. Equivalently, there must exist $\bar{\epsilon} >0$, determined by $\epsilon$, such that $\cos\angle z_1 - \cos\angle z_2\geq \bar{\epsilon}>0$, since $\cos(\cdot)$ is monotonically decreasing on $\interval{0}{\pi}$. By $z_1\in \srg_e(\sysp)$ in~\eqref{eq: def_hard_srg} and $z_2\in \srg^{\dagger}_e(\sysc)$, for all $u_1, u_2, \bar{u}_1, \bar{u}_2 \in \ltep$ such that $\angle{z_2}-\angle{z_1}\geq \epsilon$ holds, we have 
\begin{align*}
 \frac{\ininf{\Delta u_1}{\Delta y_1}_T}{\norm{\Delta u_{1}}_T\norm{\Delta y_{1}}_T} - \frac{\ininf{\Delta y_2}{\Delta u_2}_T}{\norm{\Delta u_{2}}_T\norm{\Delta y_{2}}_T} \geq \bar{\epsilon} >0
\end{align*}
for all $T>0$, where $(\Delta u_1)_T\neq 0$, $(\Delta u_2)_T\neq 0$, $(\Delta y_1)_T\neq 0$ and $(\Delta y_2)_T\neq 0$. Using $u_1=d_1+y_2$ and $u_2=y_1$, we have 
\begin{align*}
 \ininfbig{\frac{(\Delta d_{1})_T}{\norm{\Delta u_{1}}_T} +\frac{(\Delta y_{2})_T}{\norm{\Delta u_{1}}_T}}{\frac{(\Delta u_{2})_T}{\norm{\Delta u_{2}}_T}} - \ininfbig{\frac{(\Delta y_{2})_T}{\norm{\Delta y_{2}}_T}}{ \frac{(\Delta u_{2})_T}{\norm{\Delta u_{2}}_T}} \geq \bar{\epsilon} 
\end{align*}
 holds for all $T>0$.
 This gives 
 \begin{multline}\label{eq: thm1_eq3}
 \ininfbig{\frac{(\Delta d_{1})_T}{\norm{\Delta u_{1}}_T}}{\frac{(\Delta u_{2})_T}{\norm{\Delta u_{2}}_T}} \\+ 
 \ininfbig{\frac{\sbkt{\norm{\Delta y_{2}}_T-\norm{\Delta u_{1}}_T }(\Delta y_{2})_T}{\norm{\Delta u_{1}}_T\norm{\Delta y_{2}}_T}}{\frac{(\Delta u_{2})_T}{\norm{\Delta u_{2}}_T}} 
 \geq \bar{\epsilon} 
 \end{multline}
 for all $T>0$. Applying the Cauchy-Schwarz inequality to \eqref{eq: thm1_eq3} gives 
 \bex
 \frac{\norm{\Delta d_{1}}_T}{\norm{\Delta u_{1}}_T} + \frac{\abs{ \norm{\Delta y_{2}}_T -\norm{\Delta u_{1}}_T}}{\norm{\Delta u_{1}}_T} \geq \bar{\epsilon}
 \eex
 for all $T>0$.
Since $\norm{\Delta d_{1}}_T \geq \abs{\norm{\Delta u_{1}}_T - \norm{\Delta y_{2}}_T}$ for all $T>0$, it follows that
 \begin{multline*}
 \frac{\norm{\Delta d_{1}}_T}{\norm{\Delta u_{1}}_T} + \frac{\norm{\Delta d_1}_T}{\norm{\Delta u_1}_T}\geq \frac{\norm{\Delta d_{1}}_T}{\norm{\Delta u_1}_T} + \frac{\abs{ \norm{\Delta y_{2}}_T - \norm{\Delta u_{1}}_T}}{\norm{\Delta u_{1}}_T} \geq \bar{\epsilon}
 \end{multline*}
 holds for all $T>0$. This implies 
 \be\label{eq: thm1_oneside}
 \norm{\Delta u_{1}}_T \leq \frac{2}{\bar{\epsilon}} \norm{\Delta d_{1}}_T
 \ee
 for all $T>0$.
 Since by hypothesis $\sysp$ is stable and by $d_2=0$ and $\bar{d}_2=0$ so that $\Delta d_2=0$, for all $T>0$, we have 
 \begin{equation}\label{eq: thm1_eq7}
 \begin{aligned}
 &\norm{\Delta u}_T \leq \norm{\Delta u_{1}}_T + \norm{\Delta u_{2}}_T = \norm{\Delta u_{1}}_T + \norm{\Delta y_{1}}_T \\
 \leq &~(1+\inorm{\sysp}) \norm{\Delta u_{1}}_T \leq \frac{2+2\inorm{\sysp}}{\bar{\epsilon}} \norm{\Delta d_{1}}_T \eqqcolon c_3 \norm{\Delta d_{1}}_T.
 \end{aligned} 
 \end{equation} 
 
 Secondly, consider $\angle{z_1}-\angle{z_2}\geq \epsilon$. For all $u_1, u_2, \bar{u}_1, \bar{u}_2\in \ltep$ such that $\angle{z_1}-\angle{z_2}\geq \epsilon$ holds, we have
 \begin{align*}
 \frac{\ininf{\Delta y_2}{\Delta u_2}_T}{\norm{(\Delta u_{2})_T}_2\norm{\Delta y_{2}}_T} - \frac{\ininf{\Delta u_1}{\Delta y_1}_T}{\norm{\Delta u_{1}}_T\norm{\Delta y_{1}}_T} \geq \tilde{\epsilon}>0
 \end{align*}
 for all $T>0$, where $\tilde{\epsilon}$ is determined by $\epsilon$.
 Using $u_1=d_1+y_2$ and $u_2=y_1$, we then have
 \bex
 \hspace{-1mm}\ininfbig{\frac{(\Delta y_{2})_T}{\norm{\Delta y_{2}}_T}}{ \frac{(\Delta u_{2})_T}{\norm{\Delta u_{2}}_T}} -\ininfbig{\frac{(\Delta d_{1})_T}{\norm{\Delta u_{1}}_T} +\frac{(\Delta y_{2})_T}{\norm{\Delta u_{1}}_T}}{\frac{(\Delta u_{2})_T}{\norm{\Delta u_{2}}_T}} \geq \tilde{\epsilon}
 \eex
 for all $T>0$. This gives 
 \begin{multline*} 
 \ininfbig{\frac{\sbkt{\norm{\Delta u_{1}}_T-\norm{\Delta y_{2}}_T }(\Delta y_{2})_T}{\norm{\Delta u_{1}}_T\norm{\Delta y_{2}}_T}}{\frac{(\Delta u_{2})_T}{\norm{\Delta u_{2}}_T}} \\+ \ininfbig{\frac{-(\Delta d_{1})_T}{\norm{\Delta u_{1}}_T}}{\frac{(\Delta u_{2})_T}{\norm{\Delta u_{2}}_T}} 
 \geq \tilde{\epsilon}
 \end{multline*}
 for all $T>0$. Following the same reasoning as in \eqref{eq: thm1_eq3} and \eqref{eq: thm1_eq7} yields that for all $T>0$, we have
\be\label{eq: thm1_eq8} 
\norm{\Delta u}_T\leq \frac{2+2\inorm{\sysp}}{\tilde{\epsilon}} \norm{\Delta d_{1}}_T \eqqcolon c_4 \norm{\Delta d_{1}}_T.
\ee 

Combining Cases \ref{item: thm_proof_gain} and \ref{item: thm_proof_angle}, using \eqref{eq: thm1_eq5}, \eqref{eq: thm1_eq6}, \eqref{eq: thm1_eq7} and \eqref{eq: thm1_eq8} and by setting $c\coloneqq \max\bbkt{c_1, c_2, c_3, c_4}>0$, we conclude that $\norm{\Delta u}_T\leq c \norm{\Delta d_{1}}_T$ for all $u, \bar{u}\in \ltep^{2n}$ and $T>0$. Therefore, the well-posed feedback system $\gof$ with $d_2=0$ is stable.
\end{proof}

Each point $z\in \srg_e(\sysp)$ mixes both of the incremental gain and phase information contained in $\sysp$. As a consequence, an underlying idea behind the separation condition \eqref{eq: hard_SRG_separation} on the hard SRGs is \emph{separation} of the \emph{gain} and \emph{phase} information contained in $\sysp$ and that in $\sysc$. This drives us to establish the proof of Theorem~\ref{thm: hardSG} from a gain and phase perspective, which is divided into three scenarios by partitioning input-output trajectory-wise pairs in the feedback loop. They are in essence the incremental small-gain pair, large-gain pair, and small-phase pair. This proof is \emph{new} and \emph{different} from the proof of the existing soft SRG separation result \cite[Th.~2]{Chaffey:21j}.

A technical assumption in Theorem~\ref{thm: hardSG} is worthy of discussion. The feedback structure with the second external input $d_2=0$ is considered; see also \cite[p.~181]{Desoer:75}. Investigating such a structure is often sufficient \cite{Megretski:97} in comparison to exploring $\gof$ with both $d_1$ and $d_2$. Particularly, for $\gof$, where $\sysc$ is a linear bounded system, the effect of the input $d_2$ in the feedback loop can be included in that of the input $d_1$, as elaborated in \cite[Sect.~8]{Jonsson:01}. For this case, stability of $\gof$ is equivalent to that of $\gof$ with $d_2=0$. How to extend Theorem~\ref{thm: hardSG} to accommodate the second input remains open and deserves further exploration. The technical difficulty lies in generalizing the proof of the small-phase pair in Case \ref{item: thm_proof_angle}. 

In light of \eqref{eq: hard_SRG_separation}, the shortest distance between $\srg_e(\sysp)$ and $\srg^{\dagger}_e(\sysc)$ plays the role of \emph{hard-type stability margin} of a loop:
\be\label{eq: stability_margin}
\mathrm{sm}_{e}(\gof)\coloneqq \inf_{\substack{z_{1} \in \srg_e(\sysp), z_{2} \in \srg^\dagger_e(\sysc)} } \abs{ z_1 - z_2}.
\ee
It quantifies \emph{robustness} of stability of the loop in the sense that small perturbations in both $\sysp$ and $\sysc$ will not destroy the feedback stability provided that the distance always remains positive. For instance, one can fix the system $\sysc$ and allow the system $\sysp$ to be uncertain in the sense of an additive type, $\sysp=\sysp_0 + \sysg$, where $\sysp_0$ is regarded as a nominal system and $\sysg$ is uncertain whose hard SRG is known or bounded by a certain region. Then the worst case $\srg_e(\sysp)$ can be easily inferred from $\srg_e(\sysp_0)$ and $\srg_e(\sysg)$, analogously to the interconnection sum rules of soft SRGs shown in \cite[Th.~6]{Ryu:21} and \cite[Prop.~7]{Chaffey:21j}. By examining the shortest distance \eqref{eq: stability_margin} between $\srg_e^\dagger(\sysc)$ and the worst case $\srg_e(\sysp)$, one can deduce the robust stability of the uncertain feedback system $\gof$.

\subsection{Soft SRG Separation}
The second result aims at a feedback stability condition via separation of soft SRGs. In contrast to the hard-type separation in Theorem~\ref{thm: hardSG}, an extra homotopy condition on $\tau \in \interval[open left]{0}{1}$ is required. 

\begin{theorem}[Soft SRG Separation]\label{thm: softSG}
 Consider a feedback system $\gof$ with $d_2=0$, where $\sysp$ and $\sysc$ are stable. Suppose that $\goftau$ with $d_2=0$ is well-posed for all $\tau \in \interval[open left]{0}{1}$. Then, $\gof$ with $d_2=0$ is stable if
 \be\label{eq: soft_SRG_separation}
 \inf_{\substack{z_{1} \in \srg(\sysp), z_{2} \in \srg^\dagger(\tau \sysc), \tau \in \interval[open left]{0}{1}}} \abs{ z_1 - z_2} >0.
 \ee
 \end{theorem}
 \begin{proof}
 See Appendix.
 \end{proof}
 
A homotopy argument similar to that used in \cite[Th.~1]{Megretski:97} and \cite{Chaffey:21j} is adopted in both the statements and proof of Theorem~\ref{thm: softSG}. Compared with the hard-type stability margin \eqref{eq: stability_margin}, the shortest distance between $\srg(\sysp)$ and $\srg^{\dagger}(\tau\sysc)$ serves as the \emph{soft-type stability margin}, i.e.,
\bex 
\mathrm{sm}(\gof)\coloneqq \inf_{\substack{z_{1} \in \srg(\sysp), z_{2} \in \srg^\dagger(\tau \sysc), \tau \in \interval[open left]{0}{1}}} \abs{ z_1 - z_2}.
 \eex 
Robustness of the feedback stability can similarly be inferred from a positive margin $\mathrm{sm}(\gof)>0$ which can be less conservative than utilizing $\mathrm{sm}_e(\gof)>0$ due to Proposition~\ref{prop: soft_hard}. Condition \eqref{eq: soft_SRG_separation} involving $\tau$ is symmetric in $\sysp$ and $\sysc$. To be specific, instead of \eqref{eq: soft_SRG_separation}, one may also examine the distance between $\srg(\tau \sysp)$ and $\srg^{\dagger}(\sysc)$ for all $\tau \in \interval[open left]{0}{1}$ for the ease of verification. For special classes of systems like incrementally gain-bounded or incrementally positive systems, $\tau$ may be further removed and comparing $\srg(\sysp)$ with $\srg^{\dagger}(\sysc)$
 often becomes sufficient for feedback stability.

A restricted result similar to Theorem~\ref{thm: softSG} has appeared in \cite[Th.~2]{Chaffey:21j}. To clarify the difference and highlight our contribution, we draw a comparison of Theorem~\ref{thm: softSG} with \cite[Th.~2]{Chaffey:21j} in Section~\ref{sec: comparison_SRG}.

\subsection{Graphical Characterizations of Incremental Positivity Theorems and Incremental Passivity Theorems}\label{sec: passivitythm}
Based upon the soft and hard SRG separation results, we now specialize them to the celebrated incremental positivity theorem and incremental passivity theorem \cite[Sect.~VI.4]{Desoer:75}. Roughly speaking, a typical incremental positivity (resp. passivity) theorem for a \emph{negative} feedback system requires one open-loop component to be incrementally positive (resp. passive) and the other to be strictly incrementally positive (resp. passive). Such a theorem may be viewed as a direct consequence of Theorem~\ref{thm: softSG} (resp. Theorem~\ref{thm: hardSG}) as detailed below.

\begin{corollary}\label{cor: passivity_thm}
A well-posed feedback system $\boldsymbol{P}\,\#\, (-\boldsymbol{C})$ is stable with $d_2=0$ if one of the following conditions holds:
\begin{enumerate}
 \renewcommand{\theenumi}{\textup{(\roman{enumi})}}\renewcommand{\labelenumi}{\theenumi}
 \item \label{item: passivitytheorem} $\sysp$ is strictly incrementally passive and $-\sysc$ is incrementally passive.
 \item $\sysp$ is stable and strictly incrementally positive and $-\sysc$ is stable and incrementally positive;
\end{enumerate}
 \end{corollary}
\begin{proof}
Due to Proposition~\ref{prop:strictpassive}, for brevity we only show the proof of condition~\ref{item: passivitytheorem}. It follows from Definition~\ref{def: passive}\ref{item: l2e_inc_passive_def} and Proposition~\ref{prop:strictpassive} that $\srg_e(\sysp)\subset \mathcal{D}(\delta, \epsilon)$ and $\srg_e^\dagger(\sysc) \subset \clp\setminus\{0\}$, where $\mathcal{D}(\delta, \epsilon)$ is given in \eqref{eq: sectored_disk}, whereby the distance between $\srg_e(\sysp)$ and $\srg_e^\dagger(\sysc)$ is positive and hence \eqref{eq: hard_SRG_separation} is satisfied. Note that a strictly incrementally passive system always has a finite incremental gain and thus $\sysp$ is stable. The feedback stability then follows from Theorem~\ref{thm: hardSG}. 
\end{proof}

\subsection{Relations with Existing Soft SRG Separation}\label{sec: comparison_SRG}
We end this section by drawing a comparison between our main results (Theorems~\ref{thm: hardSG} and~\ref{thm: softSG}) and~\cite[Th.~2]{Chaffey:21j}. For a better comparison, we rephrase \cite[Th.~2]{Chaffey:21j} as follows and then point out its major differences from Theorems~\ref{thm: hardSG} and \ref{thm: softSG}. Given a class of systems $\mathcal{C}$, let $\overline{\mathcal{C}}$ denote a class of systems such that $\mathcal{C}\subset \overline{\mathcal{C}}$ and $\srg(\overline{\mathcal{C}})$ satisfies the \emph{chordal property} defined in \cite[p.~6070]{Chaffey:21j}.

\begin{theorem}[\hspace{1sp}{\cite[Th.~2]{Chaffey:21j}, \cite[Cor.~1]{Chaffey:24}}]\label{thm:tom}
 Consider a feedback system of $\sysp\in \mathcal{P}$ and $\sysc\in \mathcal{C}$, where $\mathcal{P}$ is a class of systems on $\lt$ with finite-incremental-gains and $\mathcal{C}$ is a class of systems on $\lt$. Then $\gof$ with $d_2=0$ is stable if there exists a class $\overline{\mathcal{C}}$ such that
 \bex
\srg^{\dagger}(\mathcal{P}) \cap \tau \srg(\overline{\mathcal{C}}) =\emptyset \quad \forall \tau\in \interval{0}{1}.
\eex
\end{theorem}

Theorem~\ref{thm:tom} can be classified as a separation result of the soft type. The technical purpose of introducing the class $\overline{\mathcal{C}}$ in Theorem~\ref{thm:tom} is to over-approximate $\srg({\mathcal{C}})$ by using $\srg(\overline{\mathcal{C}})$. In such a case, since the chordal property always holds for $\srg(\overline{\mathcal{C}})$, the soft SRG interconnection rules in \cite[Prop.~7]{Chaffey:21j} can then be adopted into the original proof of Theorem~\ref{thm:tom} which is a critical step. A refinement of Theorem~\ref{thm:tom} was recently proposed in \cite[Cor.~1]{Chaffey:24}.

Our main results have substantial contributions beyond Theorem~\ref{thm:tom}. Firstly, we have shown in Theorems~\ref{thm: hardSG} and \ref{thm: softSG} that the over-approximation assumption and the chordal property underlined in Theorem~\ref{thm:tom} are \emph{not} needed for feedback stability analysis. By contrast, our main results are proved without using any SRG interconnection rules. As a consequence, Theorem~\ref{thm: softSG} may be viewed a generalization of Theorem~\ref{thm:tom}. Secondly, note that in Theorem~\ref{thm: hardSG}, one open-loop system is allowed to be \emph{unbounded}, which further broadens the applicability of the SRG separation results for practical use. For example, a linear integrator $\frac{1}{s}I$, a commonly-seen important unbounded system on $\lt$, can now be included in feedback stability analysis in light of Theorem~\ref{thm: hardSG} (see Section~\ref{section: soft_hard_connection}). Thirdly, we introduced the notion of hard SRGs and established the hard SRG separation in Theorem~\ref{thm: hardSG} without homotopy conditions. This contribution is also new and complements the existing theory of soft SRGs. 
 
\section{Conclusion and Future Work}\label{sec:conclusion}
In this paper, first we proposed soft and hard scaled relative graphs for nonlinear systems from an input-output perspective. These graphs mix incremental gain and incremental phase information of nonlinear systems into a set of complex scalars and can fully characterize the notions of incremental positivity and incremental passivity. Novel feedback stability conditions were then developed via separation of soft SRGs and separation of hard SRGs, which were shown to recover the incremental positivity theorem and incremental passivity theorem, respectively. The proposed conditions can be perceived as a graphical statement of the classical topological graph separation of feedback systems. Finally, we made a detailed comparison between the main results and a previous soft SRG separation result. 

Future research directions include extending the soft SRG separation theorem to possibly unbounded systems as preliminarily explored in \cite{Chen:24_mtns} and connecting the hard SRG to the classical Nyquist plot.
 
\appendix 

\myproof
First, when $\tau=0$, the feedback system $\goftau$ with $d_2=0$ is stable since $\sysp$ is stable. We then consider a collection of feedback systems $\goftau$ for $\tau\in\interval[open left]{0}{1}$ as below.

\emph{Step 1:}~For all $u, \bar{u}\in \ltp^{2n}$ and $\tau \in \interval[open left]{0}{1}$ with $d_2=0$ and $\bar{d}_2=0$, show that there exists $c>0$, independent of $\tau$, such that $\norm{u-\bar{u}}_2\leq c \norm{ \sysf_{\sysp\# (\tau\sysc)} u - \sysf_{\sysp\# (\tau\sysc)} \bar{u}}_2$.
 
By hypothesis, for all $z_1\in \srg(\sysp)$, $z_2\in \srg^{\dagger}(\tau \sysc)$ and $\tau\in \interval[open left]{0}{1}$, we have $\abs{z_1-z_2}\geq \delta>0$, where $\delta>0$ is independent of $\tau$ due to condition~\eqref{eq: soft_SRG_separation}. This means at least one of the following two cases has to hold: (i) $\abs{\abs{z_1} -\abs{z_2}}\geq \epsilon$; (ii) $\abs{\angle z_1 -\angle z_2} \geq \epsilon$, where $\epsilon>0$ is independent of $\tau$ and can be determined from given $\delta$. The full proof of Step 1 follows similar reasoning to that of Theorem~\ref{thm: hardSG}. We therefore simplify it by highlighting only the major differences.
 
Case (i): $\abs{\abs{z_1} -\abs{z_2}}\geq \epsilon$. Firstly, consider $\abs{z_2}-\abs{z_1}\geq \epsilon$. By definitions~\eqref{eq: def_soft_srg} and \eqref{eq: def_inv_soft_srg}, for all $u_1, u_2, \bar{u}_1, \bar{u}_2\in \lt$ such that $\abs{z_1} \leq \abs{z_2} - \epsilon$ holds, we have 
 \begin{align} 
 \frac{\norm{\Delta y_{1}}_2}{\norm{\Delta u_{1}}_2} \frac{\norm{\Delta y_{2}}_2}{\norm{\Delta u_{2}}_2} \leq 1- \epsilon \frac{\norm{\Delta y_{2}}_2}{\norm{\Delta u_{2}}_2} < 1\quad\forall \tau \in \interval[open left]{0}{1},
 \end{align}
where the increments $\Delta u_1$, $\Delta u_2$, $\Delta y_1\coloneqq \sysp u_1 - \sysp \bar{u}_1$ and $\Delta y_2\coloneqq \tau\sysc u_2 - \tau\sysc \bar{u}_2$ with $\tau \in \interval[open left]{0}{1}$ are all nonzero. This gives that
\be\label{eq: thm2_eq5}
 \norm{u-\bar{u}}_2 \leq \frac{1 +\inorm{\sysp}}{1-\alpha}\norm{\sysf_{\sysp\# (\tau\sysc)} u - \sysf_{\sysp\# (\tau\sysc)} \bar{u}}_2 \forall \tau \in \interval[open left]{0}{1}
\ee 
for some constant $\alpha<1$ specified by the product $\alpha_1(\cdot)\alpha_2(\cdot)$, where $\alpha_1(\cdot)>0$ and $\alpha_2(\cdot)>0$ may depend on $\Delta u_1, \Delta u_2, \Delta y_1, \Delta y_2$ and $\tau$. Secondly, consider $\abs{z_1}-\abs{z_2}\geq \epsilon$. In a similar manner, for all $u_1, u_2, \bar{u}_1, \bar{u}_2\in \lt$ such that $\abs{z_1}\geq \abs{z_2} + \epsilon$ holds, we have 
\bex 
 \frac{\norm{\Delta y_{1}}_2}{\norm{\Delta u_{1}}_2} \frac{\norm{\Delta y_{2}}_2}{\norm{\Delta u_{2}}_2} \geq 1 + \epsilon\frac{ \norm{\Delta y_{2}}_2}{ \norm{\Delta u_{2}}_2} > 1\quad\forall \tau \in \interval[open left]{0}{1}.
\eex 
Accordingly, for all $\tau \in \interval[open left]{0}{1}$, we can have the following inequality:
\be\label{eq: thm2_eq6}
 \norm{u-\bar{u}}_2 \leq \frac{1 +\inorm{\sysp}}{\beta-1}\norm{\sysf_{\sysp\# (\tau\sysc)} u - \sysf_{\sysp\# (\tau\sysc)} \bar{u}}_2
\ee 
for some constant $\beta>1$ specified by the product $\beta_1(\cdot)\beta_2(\cdot)$.

Case (ii): $\abs{\angle z_1 -\angle z_2} \geq \epsilon$. Firstly, consider $\angle{z_2}-\angle{z_1}\geq \epsilon$. For all $u_1, u_2, \bar{u}_1, \bar{u}_2\in \lt$ such that $\angle{z_2}-\angle{z_1}\geq \epsilon$ holds, we have
 \begin{align*}
 \frac{\ininf{\Delta u_1}{\Delta y_1}}{\norm{\Delta u_{1}}_2\norm{\Delta y_{1}}_2} - \frac{\ininf{\Delta y_2}{\Delta u_2}}{\norm{\Delta u_{2}}_2\norm{\Delta y_{2}}_2} \geq \bar{\epsilon} >0 \quad\forall \tau \in \interval[open left]{0}{1}.
 \end{align*}
 This eventually gives us that for all $\tau \in \interval[open left]{0}{1}$,
 \begin{equation}\label{eq: thm2_eq7}
 \begin{aligned}
 \norm{u-\bar{u}}_2 \leq \frac{2+2\inorm{\sysp}}{\bar{\epsilon}} \norm{ \sysf_{\sysp\# (\tau\sysc)} u - \sysf_{\sysp\# (\tau\sysc)} \bar{u}}_2 
 \end{aligned} 
 \end{equation} 
 for all $u, \bar{u}\in \lt^{2n}$ with $u\neq v$.
 Secondly, consider $\angle{z_1}-\angle{z_2}\geq \epsilon$. For all $u_1, u_2, \bar{u}_1, \bar{u}_2\in \lt$ such that $\angle{z_1}-\angle{z_2}\geq \epsilon$ holds, we have
 \begin{align*}
 \frac{\ininf{\Delta y_2}{\Delta u_2}}{\norm{\Delta u_{2}}_2\norm{\Delta y_{2}}_2} - \frac{\ininf{\Delta u_1}{\Delta y_1}}{\norm{\Delta u_{1}}_2\norm{\Delta y_{1}}_2} \geq \tilde{\epsilon} >0\quad\forall \tau \in \interval[open left]{0}{1}.
 \end{align*}
 Then we similarly arrive at
 \be\label{eq: thm2_eq8} 
 \norm{u-\bar{u}}_2\leq \frac{2+2\inorm{\sysp}}{\tilde{\epsilon}} \norm{\sysf_{\sysp\# (\tau\sysc)} u -\sysf_{\sysp\# (\tau\sysc)} \bar{u}}_2
 \ee 
 for all $u, \bar{u}\in \lt^{2n}$ with $u\neq \bar{u}$ and $\tau \in \interval[open left]{0}{1}$.
 By combining Cases (i) and (ii) and by using \eqref{eq: thm2_eq5}, \eqref{eq: thm2_eq6}, \eqref{eq: thm2_eq7} and \eqref{eq: thm2_eq8}, there exists a constant $c >0$, independent of $\tau$, such that for all $u, \bar{u}\in \ltp^{2n}$ and $\tau \in \interval[open left]{0}{1}$, we have
 $\norm{u-\bar{u}}_2\leq c \norm{ \sysf_{\sysp\# (\tau\sysc)}u -\sysf_{\sysp\# (\tau\sysc)}\bar{u}}_2.$ 

 \emph{Step 2:}~Show that the stability of $\boldsymbol{P}\,\#\,\rbkt{\tau\boldsymbol{C}}$ with $d_2=0$ implies that of $\boldsymbol{P}\,\#\,\sbkt{(\tau+\nu)\boldsymbol{C}}$ with $d_2=0$ for all $\abs{\nu}< \mu= {1}/({c\norm{\boldsymbol{C}}_{\RI}})$. 
 
By the well-posedness, the inverse $\sbkt{\sysf_{\sysp\# (\tau\sysc)}}^{-1}$ is well defined on $\ltep^{2n}$. By hypothesis, $\sbkt{\sysf_{\sysp\# (\tau\sysc)}}^{-1}$ is incrementally bounded on $\ltp$ with $d_2=0$. Given $u\in \ltep^{2n}$ such that $d_2=0$, we define
\begin{align}\label{eq: truncation_L2e}
v_{T}\coloneqq \sbkt{\sysf_{\sysp\# (\tau\sysc)}}^{-1}\boldsymbol{\Gamma}_T\sbkt{\sysf_{\sysp\# (\tau\sysc)}u} \in \ltp^{2n}. 
\end{align}
Analogously, given $\bar{u}\in \ltep^{2n}$ such that $\bar{d}_2=0$, define $\bar{v}_T\in \lt^{2n}$. Then, by using \eqref{eq: truncation_L2e} we have $\norm{\boldsymbol{\Gamma}_T (u-\bar{u})}_2=\norm{\boldsymbol{\Gamma}_T (v_T-\bar{v}_T)}_2\leq \norm{v_T - \bar{v}_T}_2$. Following the result of Step~1 with $c>0$, we obtain
 \begin{align*}
\hspace{-1mm} &\norm{\boldsymbol{\Gamma}_T (u-\bar{u})}_2 \leq \norm{v_T - \bar{v}_T}_2\leq c\norm{\sysf_{\sysp\# (\tau\sysc)} {v}_T- \sysf_{\sysp\# (\tau\sysc)} \bar{v}_T}_2\\
&=c\norm{\boldsymbol{\Gamma}_T\sbkt{\sysf_{\sysp\# (\tau\sysc)}u - \sysf_{\sysp\# (\tau\sysc)}\bar{u}}}_2\\
& =c\left\|\boldsymbol{\Gamma}_T\sbkt{\sysf_{\sysp\#[(\tau+\nu)\sysc]} u - \sysf_{\sysp\#[(\tau+\nu)\sysc]} \bar{u} } \right. \\
 & \left. \hspace{33mm} -~\vphantom{\sysf_{\sysp\#[(\tau+\nu)\sysc]} u}\boldsymbol{\Gamma}_T\sbkt{\stbt{0}{\nu \boldsymbol{C}}{0}{0}\boldsymbol{\Gamma}_T u -\stbt{0}{\nu \boldsymbol{C}}{0}{0}\boldsymbol{\Gamma}_T \bar{u}}\right\|_2\\
 &\leq c\norm{\boldsymbol{\Gamma}_T\sbkt{\sysf_{\sysp\#[(\tau+\nu)\sysc]} u - \sysf_{\sysp\#[(\tau+\nu)\sysc]} \bar{u}}}_2 \\
 &\hspace{3.5mm} +c\norm{{\stbt{0}{\nu \boldsymbol{C}}{0}{0}\boldsymbol{\Gamma}_T u} - \stbt{0}{\nu \boldsymbol{C}}{0}{0}\boldsymbol{\Gamma}_T \bar{u}}_2\\
 &\leq c\norm{\boldsymbol{\Gamma}_T\sbkt{\sysf_{\sysp\#[(\tau+\nu)\sysc]}u - \sysf_{\sysp\#[(\tau+\nu)\sysc]} \bar{u}}}_2 \\
 & \hspace{3.5mm}+c\abs{\nu}\inorm{\boldsymbol{C}}\norm{\boldsymbol{\Gamma}_T (u-\bar{u})}_2,
 \end{align*}
where the second equality uses the causality of $\sysc$ and the last two inequalities use the incremental gain $\inorm{\boldsymbol{C}}$ and the fact
 $\norm{\boldsymbol{\Gamma}_T (\cdot)}_2$ is a nondecreasing function of $T>0$. The above inequality thus gives 
 \begin{align*}
 \norm{\boldsymbol{\Gamma}_T (u-\bar{u})}_2 \leq 
 c_0\norm{\boldsymbol{\Gamma}_T\sbkt{\sysf_{\sysp\#[(\tau+\nu)\sysc]} u - \sysf_{\sysp\#[(\tau+\nu)\sysc]} \bar{u}}}_2
 \end{align*}
 with $c_0\coloneqq \frac{c}{{1-c\abs{\nu}\inorm{\boldsymbol{C}}}}$, 
 provided that $\abs{\nu} < {1}/({c\norm{\boldsymbol{C}}_{\RI}})\eqqcolon\mu$.

\emph{Step 3:}~When $\tau=0$, $\sbkt{\sysf_{\sysp\# (\tau\sysc)}}^{-1}$ is stable with $d_2=0$ as $\boldsymbol{P}$ is stable. It has been shown in {Step 2} that $\sbkt{\sysf_{\sysp\# (\tau\sysc)}}^{-1}$ is stable with $d_2=0$ for $\tau < \mu$. Applying Step 2 iteratively and by induction, $\sbkt{\sysf_{\sysp\# (\tau\sysc)}}^{-1}$ is stable with $d_2=0$ for all $\tau \in \interval{0}{1}$. We conclude that $\gof$ with ${d_2=0}$ is stable by setting $\tau=1$. 
\endmyproof

\section*{References}

\bibliographystyle{IEEEtran}
\bibliography{HardSoftSRG}

\end{document}